\documentclass[sigconf]{acmart}

\AtBeginDocument{%
  }

\copyrightyear{2026}
\acmYear{2026}
\setcopyright{cc}
\setcctype{by}
\acmConference[ICCAD '26]{IEEE/ACM International Conference on Computer-Aided Design}{November 08--12, 2026}{San Jose, CA, USA}
\acmBooktitle{IEEE/ACM International Conference on Computer-Aided Design (ICCAD '26), November 08--12, 2026, San Jose, CA, USA}
\acmDOI{10.1145/3831252.3833973}
\acmISBN{979-8-4007-2873-0/2026/11}

\newcommand{\ignore}[1]{}
\newcommand{\fixme}[1]{\textcolor{red}{#1}}

\usepackage{algorithmic}
\usepackage{graphicx}
\usepackage{textcomp}
\usepackage{xcolor}
\usepackage{booktabs}
\usepackage{tabularx}
\usepackage{makecell}
\usepackage{array}
\usepackage{subcaption}

\begin{document}

\title[HARP: Hadamard-Domain Write-and-Verify for Noise-Robust RRAM Programming]{HARP: Hadamard-Domain Write-and-Verify \\for Noise-Robust RRAM Programming
}

\author{Ilhuan Choi, Jiwon Yoo, Yoona Lee, Yewon Jeong, Jason Jaesung Lee, and Woo-Seok Choi}
\affiliation{%
  \institution{Dept. of Electrical and Computer Engineering, ISRC, Seoul National University, South Korea}
  \country{}}
  \email{{idisian,wooseokchoi}@snu.ac.kr}

\renewcommand{\shortauthors}{Choi et al.}

\begin{abstract}
Write-and-verify (WV) is essential for multi-level RRAM programming, yet under scaled-voltage and low-SNR conditions the verify read increasingly limits mapping accuracy, convergence speed and energy.
We propose a Hadamard-domain WV framework that improves verify reliability without adding analog hardware.
\emph{HD-PV} (Hadamard-Encoded Parallel-Verify) replaces $N$ one-hot reads with $N$ orthogonal Hadamard patterns for an $N$-cell column.
Without increasing the column read count, inverse Hadamard decoding reduces uncorrelated read-noise variance by a factor of $N$ and cancels common-mode disturbances.
\emph{HARP} (Hadamard-based ADC-Energy-Reduced Parallel-Verify) further replaces full SAR conversions with lightweight comparisons because WV requires only ternary update decisions rather than full digital codes.
Across CIFAR-10, CIFAR-100, and keyword spotting under severe read noise, conventional WV loses over 20\,\% accuracy on CIFAR-10, while HD-PV and HARP limit the loss to 0.6\,\% and 1\,\% under the same memory footprint. 
Compared to conventional multi-read averaging, HD-PV and HARP achieve comparable accuracy with up to $6.1\times$ and $3.5\times$ lower latency and $6.2\times$ and $9.5\times$ better energy efficiency, respectively.
To the best of our knowledge, this is the first Hadamard-encoded verification scheme for RRAM WV.
\end{abstract}


\maketitle

\section{Introduction}

Compute-in-Memory (CiM) architectures alleviate the data-movement bottleneck of von Neumann systems, offering substantial energy and throughput benefits for edge-AI workloads~\cite{horowitz20141,hung2020challenges,9870009}.
Among embedded non-volatile memories (eNVMs), resistive random-access memory (RRAM) enables multi-level conductance storage, high density, and strong CMOS compatibility, making it well suited for analog CiM (ACiM) accelerators~\cite{chen2018neurosim,yao2020fully,wan2022compute,ielmini2025resistive}.

For edge AI, RRAM cells must be reprogrammed across DNN models and workloads, while limited device endurance constrains allowable write cycles~\cite{meng2022write,he2023prive,ielmini2025resistive}.
Beyond one-time deployment, adaptive edge learning, in-field recalibration, and growing model sizes further increase the frequency and scale of weight programming~\cite{yao2020fully,wan2022compute,zhang2023edge}.
Reliable inference therefore depends on \emph{fast and accurate conductance updates}.

Achieving such updates requires an iterative write-and-verify (WV) process.
Two factors make WV increasingly difficult as technology scales.
First, RRAM programming is nonlinear, asymmetric, and stochastic because of device mismatch and cycle-to-cycle variation~\cite{chen2018neurosim,he2019noise,yu2021rram}.
Second, scaling shrinks sensing margins and increases the influence of thermal and mismatch-induced noise in column analog-to-digital converters (ADCs)~\cite{xie2023high, li202240, tang2022low}.
As a result, under scaled conditions, the bottleneck of WV shifts from the write pulse to the verify read: each iteration estimates the current cell state from a noisy observation and decides whether to pulse up, down, or stop, so inaccurate estimates trigger incorrect updates, waste iterations, degrade mapping quality, and ultimately limit convergence.

RRAM-based ACiM systems commonly use iterative WV to program each cell until it reaches the target conductance~\cite{yao2020fully,wan2022compute,gonugondla2020swipe,meng2022write,he2023prive,yu2021rram,huang2024rwric,zhang2021efficient,ielmini2025resistive}.
While effective, this process inevitably incurs considerable latency and energy penalties, motivating extensive research to make WV more efficient. 
Prior studies reduce iteration count using MSB-first programming~\cite{gonugondla2020swipe,he2023prive}, probabilistic termination~\cite{meng2022write}, or coarse-to-fine strategies~\cite{yu2021rram,zhang2021efficient}.
However, most do not fully capture ADC behavior, even though it dominates hardware latency and energy---up to 87.8\,\% in \cite{yao2020fully} and 63.4\,\% in \cite{li202240}. 
Moreover, limited signal-to-noise ratio (SNR) under scaled supply voltages further constrains accurate WV, motivating approaches that improve verify-read quality without increasing ADC cost.

This work addresses this directly.
Our key insight is that recovering an $N$-cell column already requires $N$ read patterns.
Conventional WV uses $N$ one-hot patterns, equivalent to an identity read matrix, which is valid but statistically suboptimal.
Replacing this basis with an $N\times N$ Hadamard matrix yields the best linear unbiased cell-state estimator under i.i.d. read noise:
inverse decoding reduces uncorrelated noise variance by a factor of $N$ without additional column reads.
Moreover, because all but the first Hadamard row are balanced, it cancels common-mode disturbances for $N-1$ cells.

Building on this idea, we observe that WV requires only ternary decisions---whether each cell should move up, move down, or stop---rather than full digital codes.
HARP therefore performs Hadamard-domain compare-only verification, reducing each measurement from an $n$-step SAR conversion to one or two comparisons instead of $n$ sequential bit decisions.

The resulting framework improves circuit-level programming fidelity and translates it into higher system-level inference accuracy under identical cell precision, slice count, and array dimensions, preserving the same memory footprint.
Conversely, the same programming margin can be traded for lower memory provisioning at a fixed accuracy target.

Specifically, this paper makes the following contributions:
\begin{enumerate}
    \item \textbf{Hadamard-Encoded Parallel-Verify (HD-PV)}: We introduce, to the best of our knowledge, the first Hadamard-encoded verification scheme for RRAM WV. For an $N$-cell column, HD-PV replaces $N$ one-hot reads with $N$ Hadamard reads and reduces the effective variance of the uncorrelated read-noise component by 1/$N$ after inverse decoding. Additionally, common-mode disturbances are canceled for $N-1$ of the $N$ decoded cells. No additional analog hardware is required and the number of column read patterns does not increase. We validate the analog feasibility through post-layout simulation. 
    \item \textbf{Hadamard-based ADC-Energy-Reduced Parallel-Verify (HARP)}: We propose HARP, a Hadamard-based WV mode that trades full-resolution Hadamard decoding for decision-level Hadamard aggregation while replacing full SAR conversions with lightweight compare-only decisions.
    This approach substantially reduces ADC latency and energy during verification by up to 3.5$\times$ and 9.5$\times$, respectively, compared to multi-read averaging.
    \item \textbf{Comprehensive cross-layer evaluation}: Using device- and circuit-level modeling with system-level inference experiments on CIFAR-10, CIFAR-100, and keyword spotting (KWS), we show that the proposed methods improve mapping accuracy and convergence, maintain high inference accuracy under the same memory footprint, and outperform multi-read averaging in latency and energy.
\end{enumerate}

\begin{figure}[t]
    \centering
    \includegraphics[width=1.0\columnwidth]{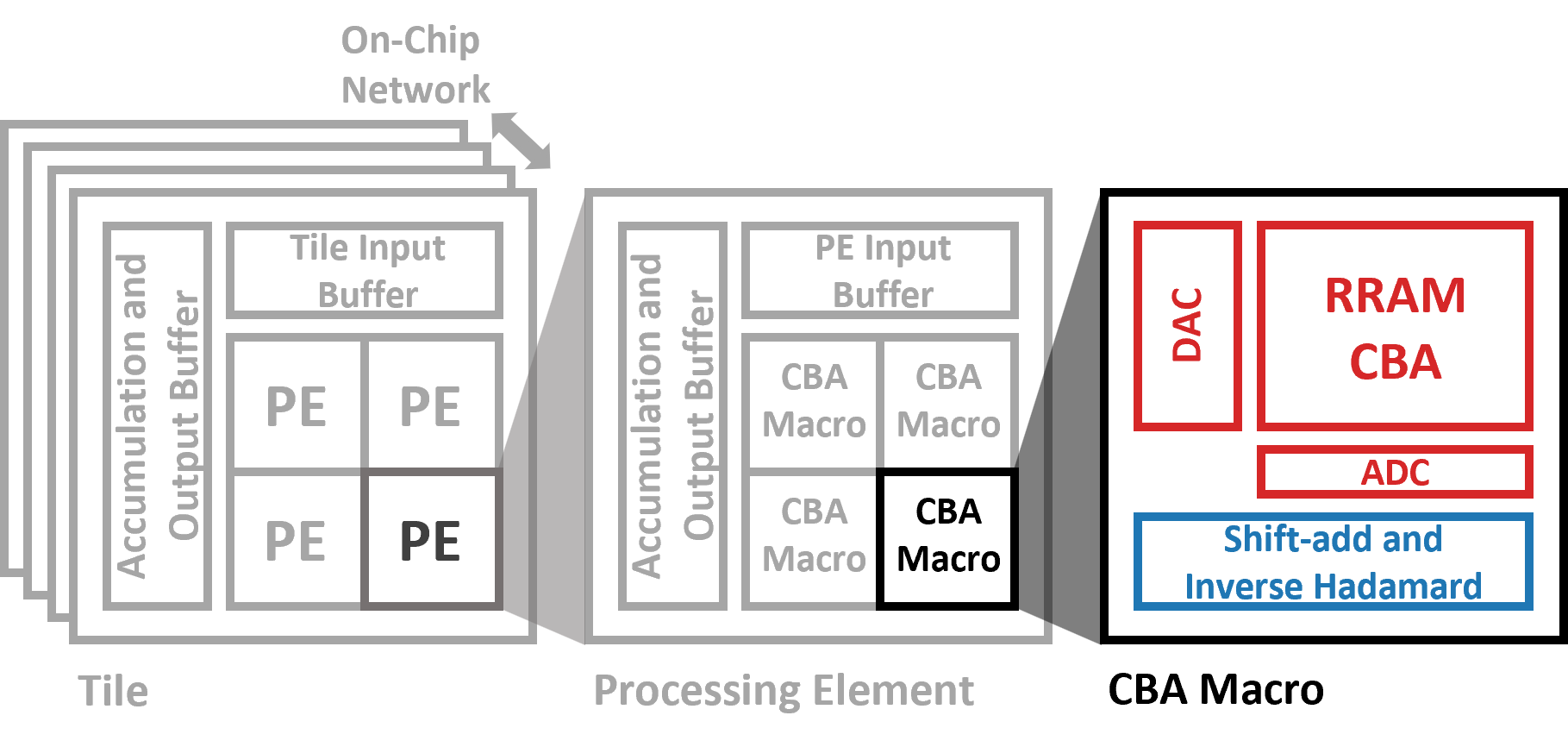}
    \vspace{-2em}
    \caption{Hierarchical organization of an ACiM chip. This work targets the highlighted CBA macro, where improving WV accuracy and efficiency can scale upward to the PE and tile levels.}
    \label{fig:ACiM_architecture}
    \vspace{-1em}
\end{figure}

\section{Background and Motivation}
\subsection{RRAM-Based ACiM Macro and Write-and-Verify}

RRAM is a non-volatile device whose conductance is tuned between a high-resistance state (HRS) and a low-resistance state (LRS) via \textit{SET}/\textit{RESET} voltage pulses. Intermediate levels enable multi-level cell (MLC) operation.

Fig.~\ref{fig:ACiM_architecture} shows the hierarchy of a typical RRAM-based ACiM chip. 
Each tile contains processing elements (PEs), and each PE integrates one or more RRAM crossbar-array (CBA) macros for parallel vector-matrix multiplication (VMM)~\cite{peng2020dnn+,andrulis2023raella}.
The CBA macro includes DACs for input encoding,
ADCs for current sensing, and shift-and-add logic for digital post-processing.
This work targets the CBA macro because macro-level WV improvement directly propagates to PE and tile-level efficiency.

\begin{figure}[t]
    \centering
    \includegraphics[width=0.9\columnwidth]{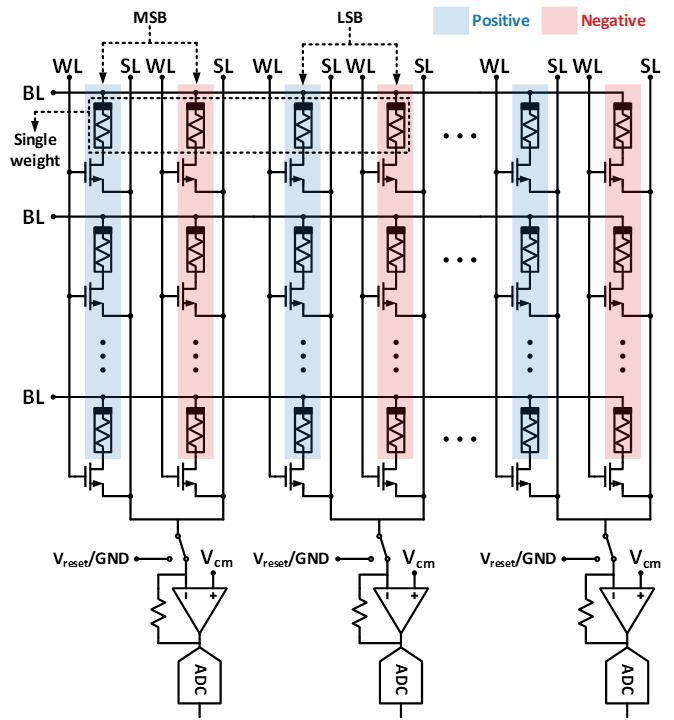}
    \vspace{-1em}
    \caption{1T-1R RRAM CBA used for ACiM. Signed weights are stored by adjacent positive and negative columns, and column currents implement VMM computation.}
    \label{fig:CBA_structure}
    \vspace{-1em}
\end{figure}

A one-transistor–one-resistor (1T–1R) RRAM CBA~\cite{meng2022write,he2023prive,yao2020fully,zhang2021efficient,correll20258} is illustrated in Fig.~\ref{fig:CBA_structure}. 
Applying input vector $\mathbf{x}$ to the bitlines (BLs) produces column currents $\mathbf{i}=\mathbf{G}\mathbf{x}$, digitized by column ADCs for VMM. 
Because RRAM stores non-negative conductance only, signed weights use adjacent positive and negative columns with effective current \( i_{\mathrm{eff}} = i_{+} - i_{-} \)~\cite{meng2022write,he2023prive,correll20258}.
Practical systems use \emph{bit-slicing}~\cite{sze2017efficient,gonugondla2020swipe,meng2022write,he2023prive}, where each weight is stored across multiple cells, 
together with \emph{bit-serial} input driving~\cite{correll20258,yao2020fully}. 
With weight precision \( B \) and \( B_C \) bits per cell, 
each weight is partitioned into $k=B/B_C$ slices
with VMM output $\mathbf{i} = \sum_{l=1}^{k} 2^{(l-1)\cdot B_C} \, \mathbf{i}^{(l)}$, where $\mathbf{i}^{(l)} = \mathbf{G}^{(l)} \mathbf{x}$
and \( \mathbf{G}^{(l)} \) is the conductance matrix of the $l$-th slice.
Because every slice must be programmed reliably, WV remains a fundamental primitive even when the inference is heavily parallelized.

\subsection{Verify-Read Error in RRAM CBA}
\label{sec:Verify-Read Error in RRAM CBA}

Accurate modeling of the verify-read error is essential for evaluating any WV improvement.
We decompose the total verify error into three components.

{\bf Programming noise.}
RRAM programming exhibits nonlinear and asymmetric \textit{SET}/\textit{RESET} behavior (see Fig.~\ref{fig:RRAM_noise}), 
along with stochastic deviations due to device mismatch (D2D) and cycle-to-cycle variations (C2C)~\cite{chen2018neurosim}. 
These effects are commonly modeled as Gaussian perturbations~\cite{he2019noise}:
\begin{equation}
w = \operatorname{clip}\!\left(w^* + n_{\text{map}},\, \text{LRS},\, \text{HRS}\right), \quad{n}_\text{map} \!\sim\! \mathcal{N}(\mathbf{0}, \sigma_\text{map}^2),
\label{eq:mapping_error}
\end{equation}
where \( w^* \) is the target, \( w \) the actual programmed conductance, and \( n_{\text{map}} \) the stochastic mapping error. The standard deviation is often expressed in a normalized form as \( \sigma_{\text{map}}/G_{\max} \), where \( G_{\max} \) denotes the maximum conductance corresponding to the LRS state~\cite{huang2024rwric,gonugondla2020swipe}.

\begin{figure}[t]
    \centering
    \includegraphics[width=1.0\columnwidth]{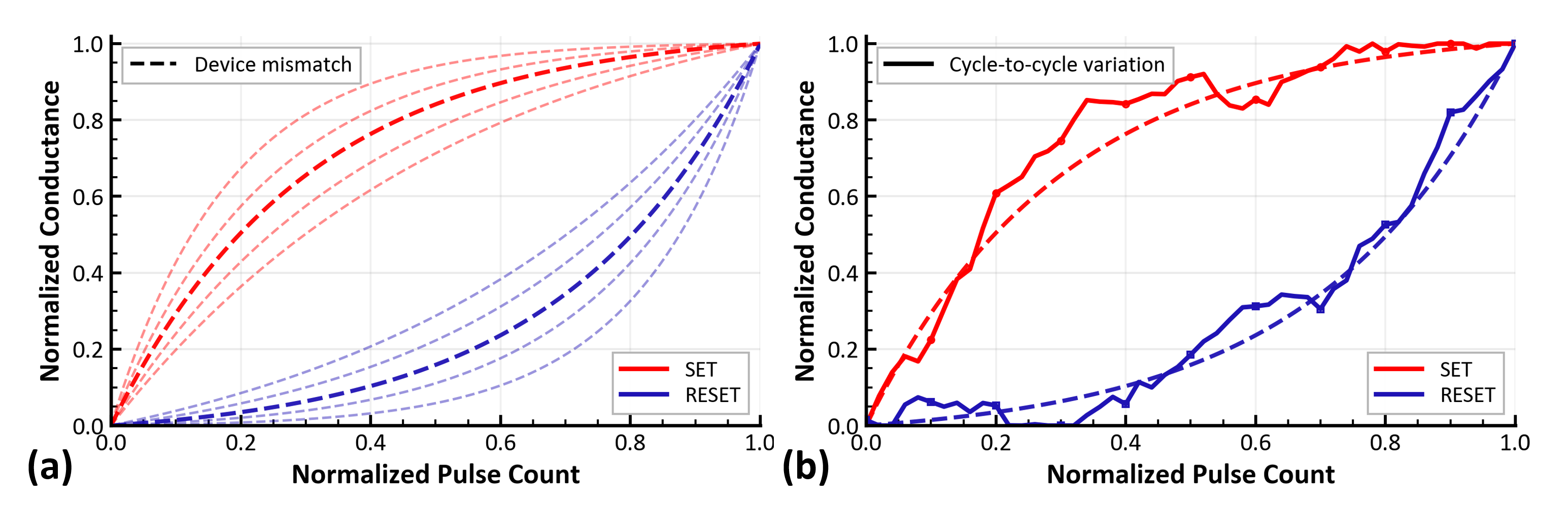}
    \vspace{-2em}    
    \caption{
    Nonlinear and asymmetric RRAM programming characteristics with (a) device mismatch and (b) cycle-to-cycle variation.
    }
    \label{fig:RRAM_noise}
    \vspace{-1em}
\end{figure}

{\bf Uncorrelated read noise.}
The dominant uncorrelated noise arises from the thermal noise in the trans-impedance amplifier (TIA) and ADC.
As ADC integration density and resolution increase, analog front-end noise becomes more significant, degrading read SNR. 
Prior measurements on MLC-RRAM arrays confirms that read noise critically limits the usable ADC resolution~\cite{shim2020two}.
Since these sources are independent across cells and across read patterns, their aggregate effect can be modeled as
\begin{equation}
{n}_{\text{uc}, i} \!\sim\! \mathcal{N}(\mathbf{0}, \sigma_\text{uc}^2), \text{  i.i.d. for each measurement $i$}.
\label{eq:uncorr_noise}
\end{equation}

{\bf Common-mode read noise.}
During verification, all $N$ measurements of a given column share the same TIA and ADC.
Any disturbance common to the analog front-end, e.g., TIA/ADC offset, reference voltage drift, or IR drop along the shared column path, appears as an additive offset that is constant across the $N$ read patterns within one verification sweep but varies independently across columns (which use separate TIAs and ADCs).
We model this as a per-column random variable:
\begin{equation}
\mu_{\text{cm}} \sim \mathcal{N}(0, \sigma_{\text{cm}}^2),  \text{  constant across measurements $i$}=1,...,N.
\label{eq:corr_noise}
\end{equation}
\ignore{
Unlike uncorrelated noise, correlated disturbances do not average out under conventional multi-read schemes.
However, as we show in Section~\ref{sec:WV}, Hadamard decoding provides a distinct advantage:
because all Hadamard rows except the first are orthogonal to the all-ones vector, common-mode disturbances that affect all cells equally are completely canceled for $N-1$ out of $N$ decoded cells.
}

The total observation for the $i$-th read pattern in a column's verification sweep is therefore:
\begin{equation}
\hat{y_i} = \mathbf{a}_i^\intercal \mathbf{w} + n_{\text{uc},i} + \mu_{\text{cm}},
\label{eq:read_noise}
\end{equation}
where $\mathbf{a}_i$ is the $i$-th row of the read matrix (identity row for conventional one-hot, Hadamard row for HD-PV), $n_{\text{uc},i}$ is i.i.d. across measurements, and $\mu_{\text{cm}}$ is shared across all $N$ measurements of that column.
Published measurement data from fabricated ACiM macros indicate that uncorrelated TIA and ADC noise is the dominant source of verify-read error at scaled supply voltages.
In the 40-nm MLC-RRAM macro of \cite{li202240}, the ADC front-end contributes the majority of the sensing uncertainty, with thermal noise from the TIA and comparator dominating at the sub-LSB level.
Similarly, the 28-nm macro of \cite{correll20258} reports that thermal noise, rather than systematic offset, limits multi-level cell discrimination.
We define the common-mode noise fraction as $\rho = \sigma_{\text{cm}}^2 / (\sigma_{\text{uc}}^2 + \sigma_{\text{cm}}^2)$.

\subsection{Optimal Read Basis for WV} \label{sec:WV}

\begin{figure}[t]
    \centering
    \includegraphics[width=1.0\columnwidth]{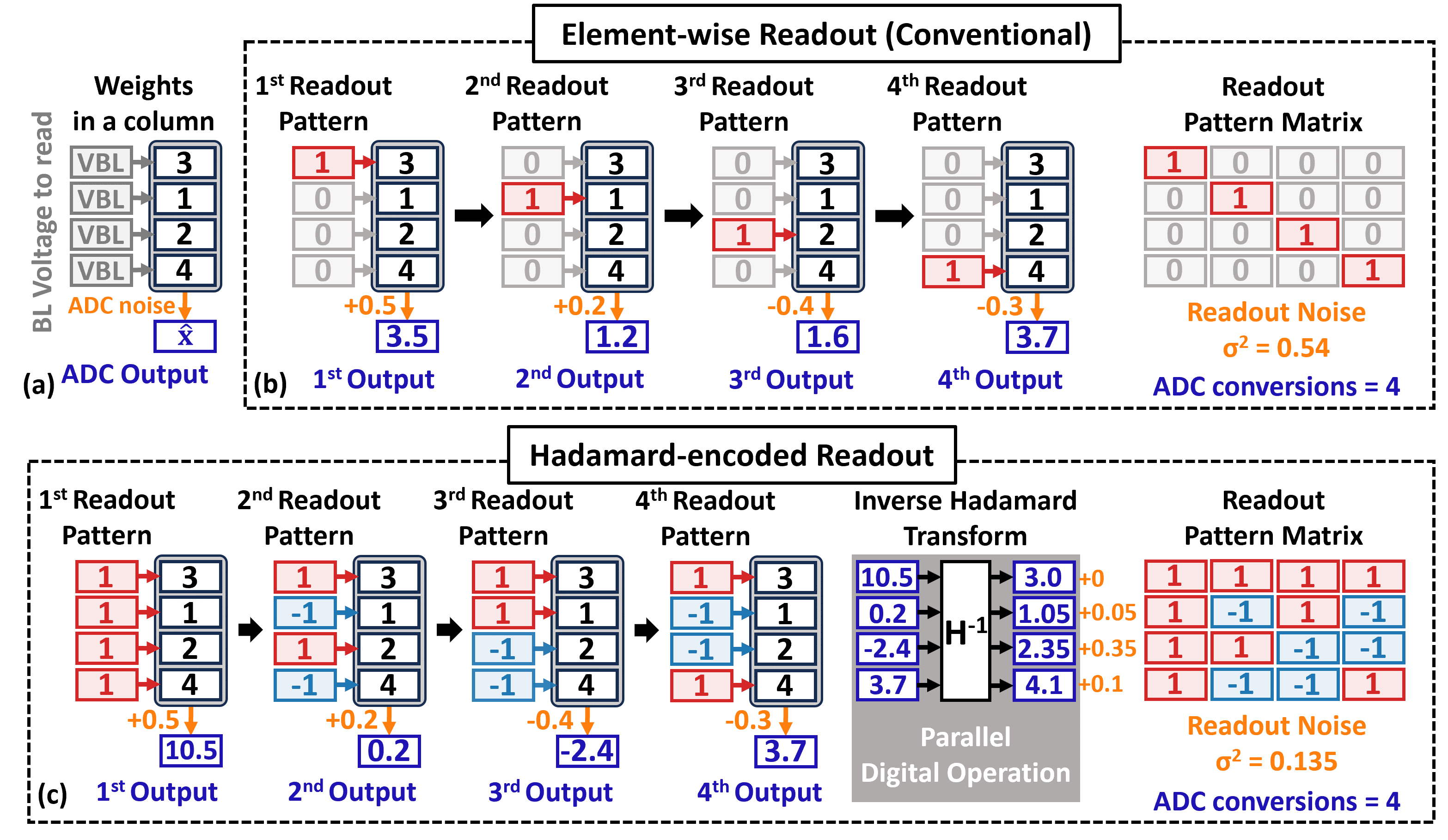}
    \vspace{-1em}
    \caption{
    Conventional element-wise (one-hot) verification and Hadamard-encoded verification. For an $N$-cell column, both use $N$ read patterns; Hadamard decoding averages the uncorrelated noise across measurements.
    }
    \label{fig:hadamard_snr}
    \vspace{-1em}
\end{figure}

At each WV iteration $t$, the verifier observes $\hat{y}_{i}$ from \eqref{eq:read_noise} and must make a ternary decision $D_i \in \{\text{SET},\text{RESET},\text{STOP}\}$.
Conventional element-wise (one-hot) verification uses the identity read matrix, providing a maximum-likelihood estimate from a single noisy observation per cell.
This is statistically inefficient: the $N$ observations used to read $N$ cells contain no redundancy.

A better approach chooses a read matrix $\mathbf{A} \in \{-1, +1\}^{N\times N}$ (matching the BL driving constraint) that minimizes the estimation error covariance $\mathrm{Cov}(\hat{\mathbf{x}}) = \sigma^2(\mathbf{A}^\intercal\mathbf{A})^{-1}$.

\begin{proposition} \label{thm:Hadamard}
{\bf (Optimality of Hadamard Encoding)}. Among all $\pm1$ measurement matrices $\mathbf{A} \in \{-1, +1\}^{N\times N}$, the per-cell estimation variance $\mathrm{Var}(\hat{x_i})$ due to uncorrelated read noise is minimized when $\mathbf{A}^\intercal\mathbf{A} = N\mathbf{I}$, yielding $\mathrm{Var}(\hat{x_i})=\sigma_{\text{uc}}^2/N$ for all $i$. 
Any Hadamard matrix achieves this bound. The identity matrix yields $\mathrm{Var}(\hat{x_i})=\sigma_{\text{uc}}^2$, which is $N\times$ larger.
\end{proposition}
\begin{proof}
Each row of $\mathbf{A}$ has $\pm1$ entries, so $\mathrm{tr}(\mathbf{A}^\intercal\mathbf{A})=N^2$. Let $\lambda_1, \dots,\lambda_N$ be the eigenvalues of $\mathbf{A}^\intercal\mathbf{A}$. 
Since $\mathbf{A}^\intercal\mathbf{A}$ is positive definite for any invertible $\mathbf{A}$, all $\lambda_i>0$.
By the Cauchy-Schwarz inequality, $\mathrm{tr}((\mathbf{A}^\intercal\mathbf{A})^{-1})=\sum_i(1/\lambda_i) \geq N^2/\sum_i\lambda_i=1$, with equality iff all $\lambda_i=N$, i.e., $\mathbf{A}^\intercal\mathbf{A} = N\mathbf{I}$, which is the property of a Hadamard matrix~\cite{khashaba2019low,roh2023context,harwit2012hadamard,hotelling1944some}.
\end{proof}

Fig.~\ref{fig:hadamard_snr} compares the readout noise between conventional element-wise and Hadamard-encoded verification.
The $\pm1$ constraint is not artificial---it matches the physical BL driving constraint exactly in the existing macro (Section~\ref{sec:HD-PV-in-RRAM}).
Under Hadamard encoding, the column is measured as
\begin{equation}
    \hat{\mathbf{y}} = \mathbf{H}\mathbf{x} + \mathbf{n}, \quad \mathbf{H}^\intercal \mathbf{H} = N\mathbf{I}.
\label{eq:hadamard_read}
\end{equation}
Inverse Hadamard decoding yields
\begin{equation}
    \hat{\mathbf{x}}_{\text{H}}
    = \frac{1}{N}\mathbf{H}^\intercal \hat{\mathbf{y}}
    = \mathbf{x} + \frac{1}{N}\mathbf{H}^\intercal \mathbf{n},
\end{equation}
with decoded noise variance $\sigma_\text{uc}^2/N$---the optimal bound from Proposition~\ref{thm:Hadamard}.

{\bf Common-mode cancellation.}
In addition to uncorrelated noise reduction, Hadamard decoding cancels common-mode disturbances.
Since $\mu_\text{cm}$ is constant across all $N$ measurements, the $j$-th decoded cell receives $(1/N)\sum_i H_{ji}\mu_\text{cm} = (\mu_\text{cm}/N)\sum_i H_{ji}$.

For $j=1$ (all +1 row), $\sum_i H_{ji} = N$, so the first cell sees the full offset $\mu_\text{cm}$. For $j\geq 2$, the balanced rows satisfy $\sum_i H_{ji} = 0$, so:
\begin{equation}
    (1/N) \mathbf{H}^\intercal (\mu_\text{cm} \cdot \mathbf{1})
    = \mu_\text{cm} \cdot \mathbf{e_1},
\label{eq:CM_cancel}
\end{equation}
where $\mathbf{e_1}$ is the first standard basis vector, and the remaining $N-1$ cells are completely free of common-mode noise.
Hadamard encoding thus provides a dual benefit:
(i) $1/N$ variance reduction for uncorrelated noise across all $N$ cells (Proposition~\ref{thm:Hadamard}), and (ii) complete cancellation of common-mode disturbances for $N-1$ out of $N$ cells.
Conventional multi-read averaging~\cite{joshi2020accurate,wan2022compute,shim2020two} reduces uncorrelated noise (at $M\times$ the read cost) but passes $\mu_\text{cm}$ to every cell unchanged, since repeated reads of the same column share the same TIA/ADC and therefore the same $\mu_\text{cm}$ realization.

The important fairness point is that both conventional and Hadamard verification use $N$ read patterns---Hadamard encoding changes the read basis, not the column read count.
Hadamard verification helps most when
(i) uncorrelated noise is significant ($1/N$ variance reduction), (ii) common-mode disturbances are present (cancellation for $N-1$ cells), and (iii) the column length $N$ is large.
All three conditions are increasingly satisfied as technology scales and array sizes grow.

\section{HD-PV with Column-Wise Write-and-Verify}
This section introduces: 
1) a column-wise WV scheme for parallel cell updates, and 
2) the first application of Hadamard encoding within the WV loop, enabled by reusing existing BL/SL drivers and SAR ADCs with no additional analog overhead.
We summarize the notation and assumptions:
\begin{itemize}
    \item Target, programmed, and read conductance vectors: $\mathbf{w}^*$, $\mathbf{w}$, $\hat{\mathbf{w}}$ per \eqref{eq:mapping_error} and \eqref{eq:read_noise}.
    \item TIA common-mode: $V_{\mathrm{cm}} = \tfrac{1}{2}V_{\mathrm{DD}}$ (see Fig.~\ref{fig:CBA_structure}). SL is disconnected from ADC during SET/RESET, and reconnected during verification in WV or inference.
    \item Initial writes from HRS (zero weight) use coarse SET pulses (higher voltage); fine updates use $\approx$0.25~LSB/step SET/RESET pulses~\cite{yu2021rram, zhang2021efficient}.
\end{itemize}

\subsection{Column-Wise Write-and-Verify Flow}
\label{sec:column-wise wv}

\begin{figure}[t]
    \centering
    \includegraphics[width=1.0\columnwidth]{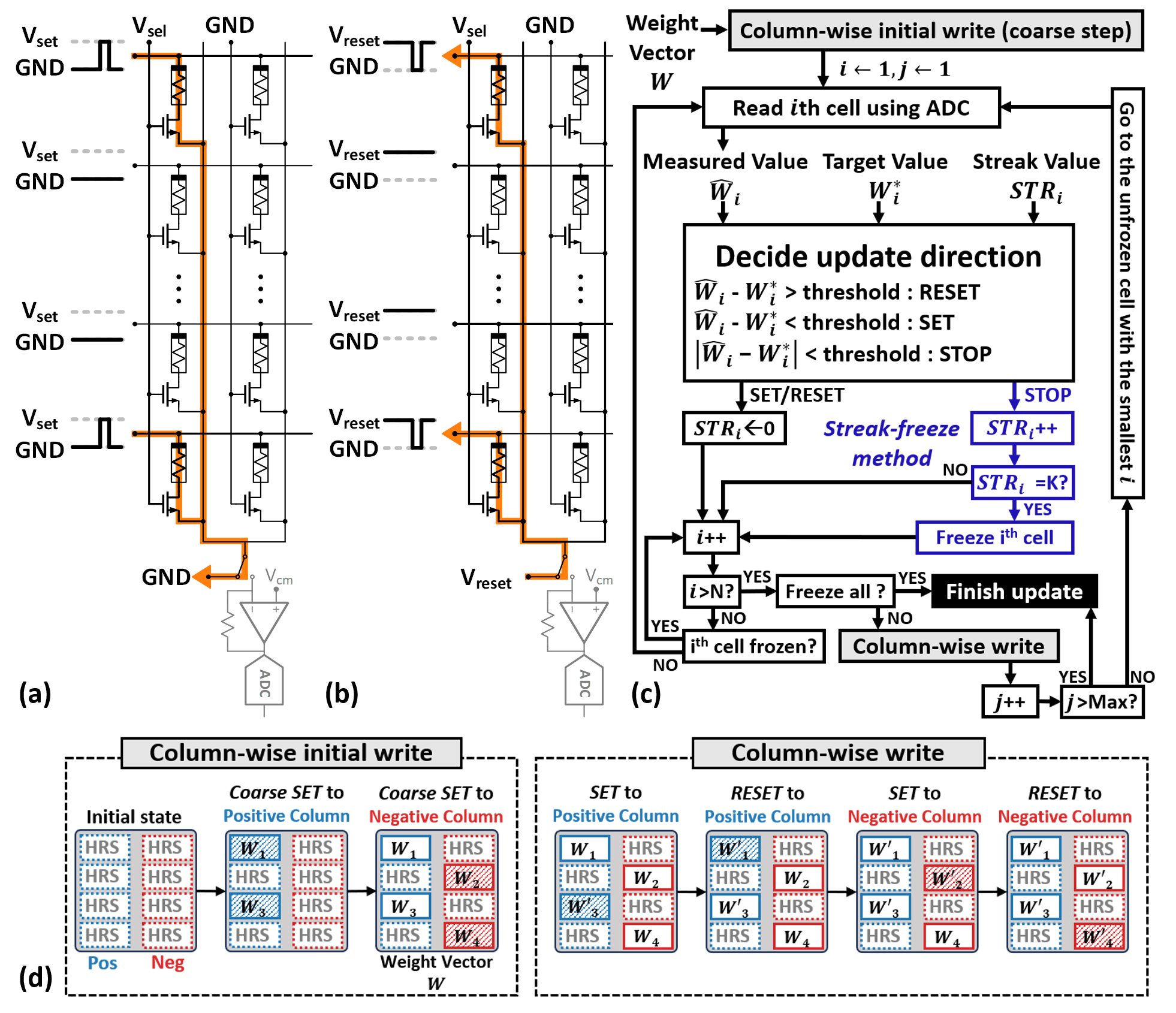}
    \vspace{-1em}
    \caption{Column-wise WV used by all methods: (a), (b) parallel SET/RESET operations, (c) conventional decision flow, and (d) signed mapping sequence of a weight vector. The hatched cells indicate the cells updated at each step.}
    \label{fig:col_wv}
    \vspace{-1em}
\end{figure}

We adopt a column-wise WV backend so that multiple cells are updated in parallel (see Fig.~\ref{fig:col_wv}).
Each verification sweep classifies every cell as SET, RESET, or STOP based on its deviation from the target using ADC (threshold: 0.5~LSB), i.e., RESET if \(w_i - w_i^* > 0.5\text{LSB}\), SET if \(w_i - w_i^* < -0.5\text{LSB}\), and STOP otherwise.
A streak counter freezes a cell only after $K$ consecutive within-threshold reads, preventing premature freezing from noisy observations.

Fig.~\ref{fig:col_wv}(a)-(b) shows the column-wise write operation. 
Once the required pulse counts are determined, SET/RESET updates are applied simultaneously to all target cells by selecting the corresponding WL with \(V_{\mathrm{sel}}\) and biasing the SL. 
During SET, BLs receive \(V_{\mathrm{set}}\) pulses while SL is grounded; during RESET, BLs receive GND pulses while SL is biased at \(V_{\mathrm{reset}}\). 
Unlike sequential cell-by-cell programming, a single write operation updates multiple cells concurrently, providing a speedup, and the latency of a column update is determined by the most demanding cell in that phase.
Fig.~\ref{fig:col_wv}(d) illustrates the signed mapping sequence, where each weight is stored by a positive/negative pair and one cell in the pair remains at HRS to encode zero.
The update consists of four phases: SET/RESET to the positive column and SET/RESET to the negative column. Frozen cells skip the update.

\subsection{Hadamard-Encoding in an RRAM CBA}
\label{sec:HD-PV-in-RRAM}

\begin{figure}[t]
    \centering
    \includegraphics[width=1.0\columnwidth]{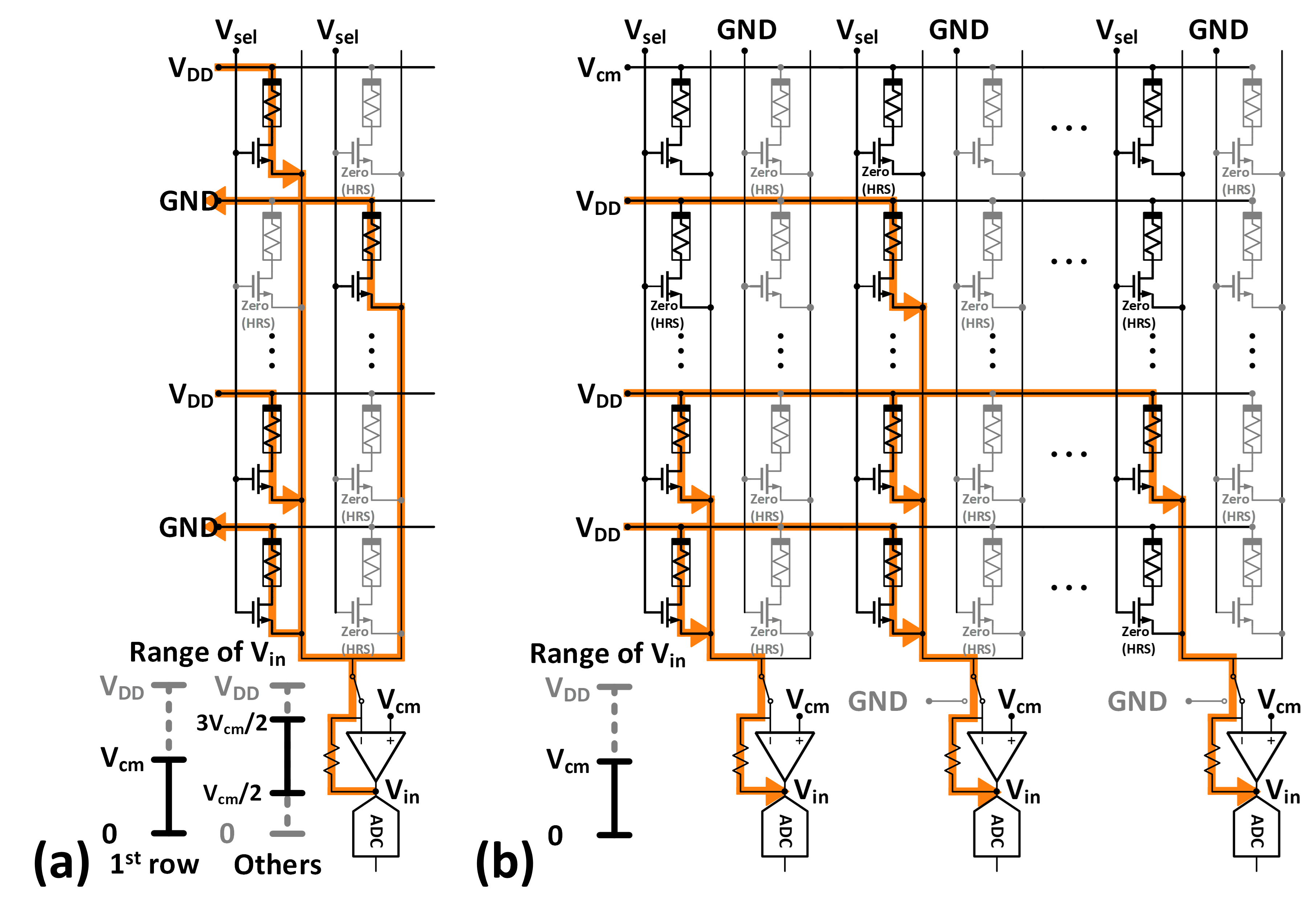}
    \vspace{-2em}
    \caption{CBA architecture supporting Hadamard-encoded verification: 
    (a) implementing $\{\pm1\}$ Hadamard entries via BL/SL driving during WV; 
    (b) VMM with bit-serial input and signed weights during inference.}
    \label{fig:hadamard_arch}
    \vspace{-1em}
\end{figure}

HD-PV replaces the conventional one-hot verify reads with Hadamard-encoded sensing using existing macro hardware.

\ignore{
To perform WV, a voltage pattern is applied to the bit-lines and the resulting signals are digitized by an ADC, whose output includes conversion noise (see Fig.~\ref{fig:hadamard_snr}(a).)
Conventional WV reads each cell individually, containing full analog noise at every measurement (see Fig.~\ref{fig:hadamard_snr}(b)).
Multi-read averaging~\cite{joshi2020accurate} improves SNR---averaging \(M\) repeated reads reduces noise power by \(1/M\) (\(\sigma_{\text{avg}}^2 = \sigma_{\text{single}}^2/M\))---but increases latency and energy by \(M\times\).

For the first time in WV, we apply Hadamard-encoded sensing such that the crossbar-array output under the $i$-th Hadamard pattern becomes \( \hat{y_i} = \mathbf{H}_i^\intercal \mathbf{w} + n_i \) ($i\in\{1,...,N\}$), where \(N\) denotes the number of cells in a column, and \( n_i \) are independent and identically distributed read-noise samples (see~\eqref{eq:hadamard_read}). 
Because noise terms are uncorrelated across the $N$ read operations, the inverse Hadamard transform averages them out during decoding and reduces noise variance to $\sigma^2/N$ as in \eqref{eq:SNR_improvement} (see Fig.~\ref{fig:hadamard_snr}(c)).
This yields a substantial SNR improvement without multiple reads, making Hadamard encoding far more efficient than multi-read averaging.
}

{\bf Bitline encoding.}
Hadamard-encoding is performed through controlled BL/SL voltages as shown in Fig.~\ref{fig:hadamard_arch}(a).
When the SL is held at \(V_{\mathrm{cm}} = \tfrac{1}{2}V_{\mathrm{DD}}\) by the TIA, BLs driven to \(V_{\mathrm{DD}}\) encode \(+1\), while BLs driven to GND encode \(-1\). 
Cells under a \(+1\) entry drive current from BL to SL, whereas cells under a \(-1\) entry drive current from SL to BL. 
For the first Hadamard row (all \(+1\)), the net current drives the TIA output below \(V_{\mathrm{cm}}\), so \(V_{\mathrm{in}}\) lies in \([0,\,V_{\mathrm{cm}}]\), which is the same range used for inference. 
For the remaining Hadamard rows (balanced \(+1/-1\)), the net current can be either positive or negative, placing \(V_{\mathrm{in}}\) within \([\tfrac{1}{2}V_{\mathrm{cm}},\,\tfrac{3}{2}V_{\mathrm{cm}}]\). 
Since signed weights are stored using adjacent column pairs for positive and negative components, with one cell in each pair kept in HRS to represent zero,
simultaneous selection and readout of both columns can be performed, as shown in Fig.~\ref{fig:hadamard_arch}(a). 
During bit-serial inference, input '1' drives \(V_{\mathrm{DD}}\) to the BL and input '0' drives \(V_{\mathrm{cm}}\) (see Fig.~\ref{fig:hadamard_arch}(b)).

\begin{figure}[t]
    \centering
    \includegraphics[width=0.95\columnwidth]{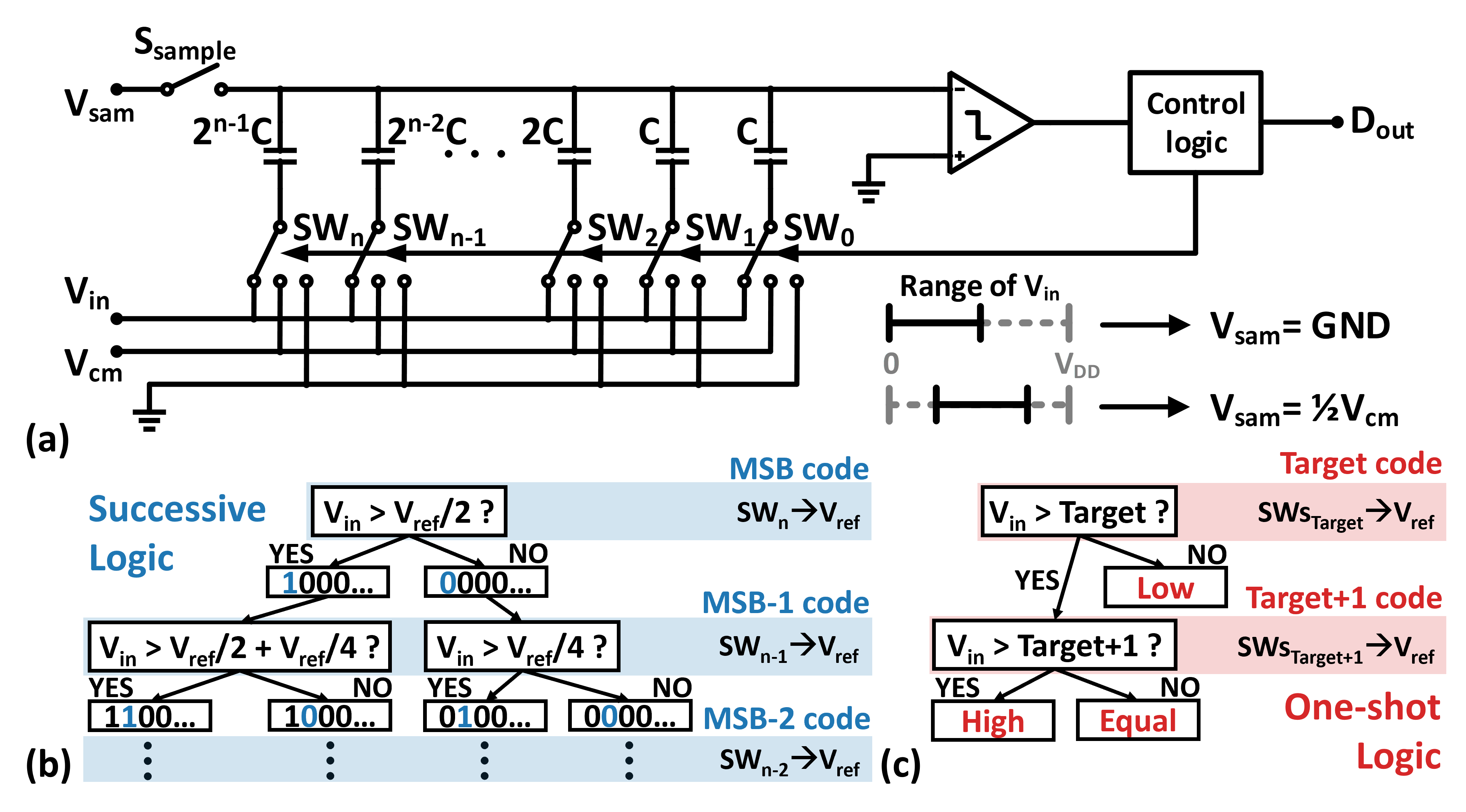}
    \vspace{-1em}
    \caption{
    Operation of an \(n\)-bit SAR ADC: 
    (b) full conversion using \textit{SAR logic} (requiring $n$ sequential comparisons), and 
    (c) lightweight \textit{compare logic} performing a simple magnitude check against the target (requiring one or two comparisons).
    }
    \label{fig:multi_read}
\end{figure}

\begin{figure}[t]
    \centering
    \includegraphics[width=1.0\columnwidth]{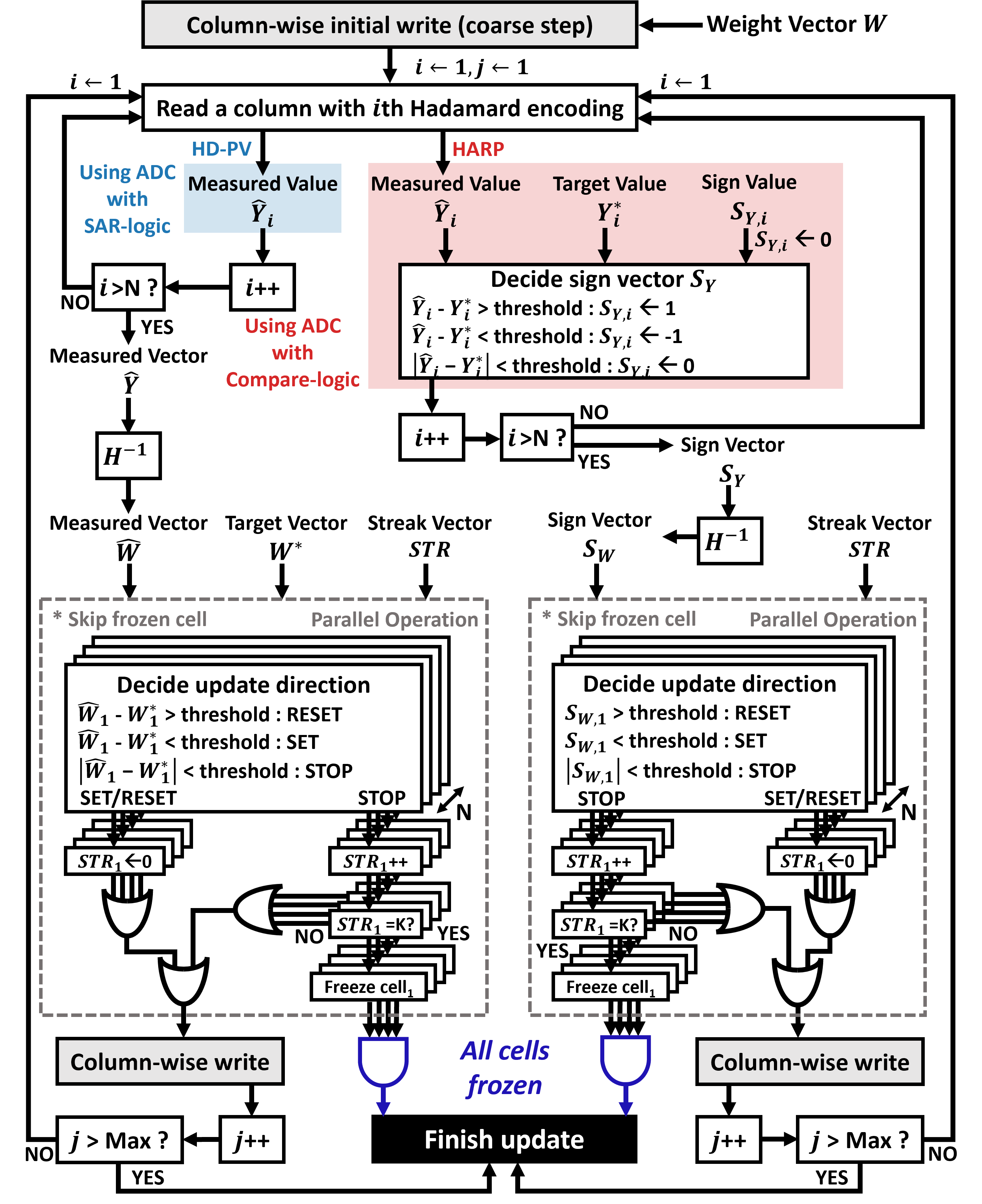}
    \vspace{-1em}
    \caption{HD-PV and HARP algorithms. HD-PV reconstructs full-valued cell estimates, whereas HARP decodes ternary Hadamard-domain comparison outcomes before generating parallel cell updates.}
    \label{fig:hd_pv}
    \vspace{-1em}
\end{figure}

{\bf ADC sampling reference.}
The same SAR ADC can accommodate both the first-row and balanced-row ranges by switching the sampling reference as shown in Fig.~\ref{fig:multi_read}(a).
Setting \(V_{\mathrm{sam}}\) = GND allows the ADC to naturally cover the first-row range (\([0,\,V_{\mathrm{cm}}]\)). 
Setting \(V_{\mathrm{sam}}\) = \(\tfrac{1}{2}V_{\mathrm{cm}}\) enables adjustment of the input range around $V_\mathrm{cm}$ (\([\tfrac{1}{2}V_{\mathrm{cm}},\,\tfrac{3}{2}V_{\mathrm{cm}}]\)) for the balanced rows. 
By simply switching \(V_{\mathrm{sam}}\) between these two reference levels, the existing SAR ADC can accommodate the full dynamic range of Hadamard-encoded readout without increasing resolution or modifying the analog front-end, enabling efficient Hadamard-based verification on the same hardware used for inference.

{\bf Digital decoding.}
Inverse Hadamard decoding is implemented digitally. 
Since bit-sliced inference already includes shift-and-adder logic to combine partial sums, the periphery naturally provides parallel adders that can accumulate the Hadamard measurements and reconstruct all $N$ decoded cell estimates.
Measurements from successive Hadamard patterns are streamed into these adders, enabling all $N$ inverse Hadamard outputs to be reconstructed in parallel.
The latency and energy overhead of this digital step is quantified in Section~\ref{sec:simulation_results}.

\subsection{HD-PV Algorithm}
\label{sec:HD-PV}


The overall HD-PV flow is illustrated in Fig.~\ref{fig:hd_pv}. 
A column is read through $N$ Hadamard patterns to obtain \(\hat{\mathbf{y}}\), which is then decoded via inverse Hadamard transform using digital adders. The decoded estimate \(\hat{\mathbf{w}}\) is compared with \(\mathbf{w}^*\), and SET/RESET/STOP decisions are generated in parallel for all cells. Column-wise writing is then applied as in Fig.~\ref{fig:col_wv}. Each cell is frozen once its streak counter reaches $K$, and HD-PV terminates when all cells are frozen or the maximum iteration count is reached.
Compared to conventional one-hot (single-cell) verify reads, HD-PV improves the quality of each verification sweep without increasing the number of column reads per sweep.


\section{HARP for Reduced WV Energy}
\subsection{ADC Costs in HD-PV}

HD-PV improves verify reliability, but it still performs a full-resolution SAR conversion for every Hadamard measurement. 
This is expensive: an \(n\)-bit SAR ADC executes an \(n\)-step binary search (see Fig.~\ref{fig:multi_read}) and 
incurs \(n\)-cycle latency and substantial switching energy in both the capacitor array and SAR logic. 
In practical ACiM macros, where WV must be repeated many times across many columns, this per-read cost dominates latency and energy.

The key insight behind HARP is that WV is fundamentally a classification problem, not an estimation problem.
The verifier needs to assign each cell to one of three categories \{SET, RESET, STOP\}---it does not need the full digital code.
If the Hadamard-domain measurement can be compared directly with the Hadamard-domain target, the verifier can make this classification with far less ADC work.

\subsection{HARP with One-Shot Target Comparison}

Fig.~\ref{fig:multi_read}(c) illustrates a one-shot target-comparison mode that can be incorporated into standard SAR ADCs. 
In a standard SAR conversion, the ADC samples $V_\mathrm{in}-V_\mathrm{sam}$ onto a binary-weighted capacitor array and resolves the output bit-by-bit.
In the HARP mode, the ADC sets all capacitor-array switches to the target code in a single step and performs one comparison of the input voltage \(V_{\mathrm{in}}\) against the target level. 
If necessary, the ADC performs one more comparison against the adjacent code (target+1), yielding a ternary outcome: \emph{Low}, \emph{High}, or \emph{Equal}. 
This reduces each Hadamard measurement to one or two comparisons instead of $n$ sequential ones. 
This one-shot comparison is compatible with one-hot verification but not with multi-read averaging or HD-PV, which rely on high-resolution conversion results. 


For the $i$-th Hadamard row $\mathbf{H}_i^\intercal$, the noisy measurement is 
\begin{equation}
\hat{{y}}_i = \mathbf{H}_i^\intercal \mathbf{w} + {n}_i, \qquad {n}_i \!\sim\! \mathcal{N}(\mathbf{0}, \sigma^2),
\label{eq:harp_read}
\end{equation}
where $\mathbf{w}$ is the current weight vector and $n_i$ represents analog read noise. 
Instead of digitizing $\hat{y}_i$ with an $n$-step SAR sequence, HARP compares it directly against the ideal Hadamard-domain target $y_i^* = \mathbf{H}_i^\intercal \mathbf{w}^*$ and generates a ternary sign value:
\begin{equation}
s_{y,i} = 
\begin{cases}
\!\operatorname{sign}(\hat{y}_i - y_i^*), & |\hat{y}_i - y_i^*| > 0.5\mathrm{LSB},\\
0, & \text{otherwise.}
\end{cases}
\label{eq:harp_decision}
\end{equation}
The resulting Hadamard-domain sign vector $\mathbf{s_y} = [s_{y,1}, s_{y,2}, \ldots, s_{y,N}]^\intercal$ 
is then decoded back to the cell domain via inverse Hadamard transform:
\begin{equation}
\mathbf{s_w} 
= \mathbf{H}^{-1}\mathbf{s_y}
= \tfrac{1}{N}\mathbf{H}^\intercal\mathbf{s_y},
\label{eq:harp_inverse}
\end{equation}
where each element $s_{w,i}$ serves as an aggregated update indicator for the $i$-th cell. 
A cell-domain threshold $\tau_w$ determines the update direction:
\begin{equation}
D_i = 
\begin{cases}
\text{RESET}, & s_{w,i} > \tau_w, \\
\text{SET},   & s_{w,i} < -\tau_w, \\
\text{STOP},  & |s_{w,i}| \le \tau_w.
\end{cases}
\label{eq:harp_decision_sw}
\end{equation}
These updates are executed in parallel using the column-wise write procedure (Fig.~\ref{fig:col_wv}), and the same streak-based termination criterion from Section~\ref{sec:column-wise wv} is applied to ensure robust convergence.

Compared with HD-PV, HARP offers two main advantages.
First, replacing full SAR conversions with one-shot comparisons greatly reduces ADC energy, which dominates the WV cost.
Second, applying the inverse Hadamard transform to ternary sign data rather than multi-bit ADC outputs reduces digital switching activity and power. 
Overall, HARP trades full-resolution estimation for Hadamard-domain decision aggregation, delivering fast and more energy-efficient WV convergence.

\begin{table}[t]
\centering
\caption{Device and circuit parameters used in evaluation.}
\label{tab:sim_param}
\renewcommand{\arraystretch}{1.15}
\setlength{\tabcolsep}{3pt}
\small
\resizebox{\columnwidth}{!}{%
\begin{tabular}{lll}
\toprule[1.5pt]
\multicolumn{3}{c}{\textbf{Device Parameters}} \\
\midrule
RRAM conductance & \multicolumn{2}{c}{0–13~\(\mu\)S} \\
SET/RESET pulse & 2~V / 100~ns & 1~step~per~pulse (50 iterations total) \\
Coarse SET pulse & 4~V / 100~ns & 5~steps~per~pulse (10 iterations total)\\
\midrule[1.5pt]
\multicolumn{3}{c}{\textbf{Circuit Parameters}} \\
\midrule
Supply voltage & \multicolumn{2}{c}{0.9~V} \\
Read pulse width & \multicolumn{2}{c}{32~ns} \\
TIA + ADC latency & \multicolumn{2}{c}{45–50~ns (SAR logic) / 30~ns (compare logic)} \\
TIA + ADC energy & \multicolumn{2}{c}{1.44–2.7~pJ + 1.8–32~pJ} \\
Inverse Hadamard adder latency & \multicolumn{2}{c}{5~ns} \\
Inverse Hadamard adder energy & \multicolumn{2}{c}{0.8–1.0~pJ (HD-PV) / 0.2~pJ (HARP)} \\
\bottomrule[1.5pt]
\end{tabular}%
}
\vspace{-1em}
\end{table}

\section{Evaluation}
\label{sec:simulation_results}

{\bf Setup and noise calibration.}
Table~\ref{tab:sim_param} summarizes the device and circuit parameters.
RRAM conductance characteristics and programming behavior are modeled using NeuroSim~\cite{chen2018neurosim}, while the noise parameters are calibrated against measured RRAM-CIM macros~\cite{correll20258,wu2023bulk,spetalnick20232}.
Programming and readout conditions follow prior measurements~\cite{correll20258,yu2021rram}.
Pulse amplitudes and widths use typical in-memory operating ranges, and inverse Hadamard decoding latency and energy are obtained from 28-nm synthesized circuits.

The fabricated 0--13~$\mu$S, 3-bit MLC RRAM macro in~\cite{correll20258} reports a 0.73~LSB read-current error including SAR-ADC noise, motivating our 0--0.8~LSB sweep and the use of 0.7~LSB as the severe-noise case.
Other measured macros report comparable programming variation, supporting the $\sigma_{\mathrm{map}}/G_{\max}=0.05$--0.10 setting~\cite{wu2023bulk,spetalnick20232}.

{\bf Post-layout validation.}
To validate feasibility under realistic parasitics, we performed post-layout simulations of a 32$\times$32 1T1R array implemented using the NCSU 45-nm FreePDK RRAM add-on~\cite{giacomin2018resistive} and Stanford RRAM model.
The $+1/-1$ read-current difference remained below 0.01~LSB across all 3-bit MLC levels and below 0.05~LSB for the highest-current ResNet-20 column.
At the average column current, Hadamard-pattern-dependent IR drop reduced the target current by only 0--0.08~LSB, with marginal impact on the RMS error and iteration count of HD-PV and HARP.

{\bf Baselines and comparison methodology.}
Unless noted, all methods use the same array size, cell precision, and column-wise write backend.
The primary baseline, \emph{column-wise single-cell (CW-SC) WV}, uses one-hot reads and the same compare-only ADC mode as HARP, isolating the read basis effect.
Multi-read averaging uses full SAR conversions because averaging requires full-resolution values.
For each $N$-cell column sweep, CW-SC, HD-PV, and HARP use $N$ read patterns, whereas $M$-read averaging requires $M \times N$ ADC conversions---a fundamentally higher cost.

\begin{figure}[t]
    \centering
    \includegraphics[width=1.0\columnwidth]{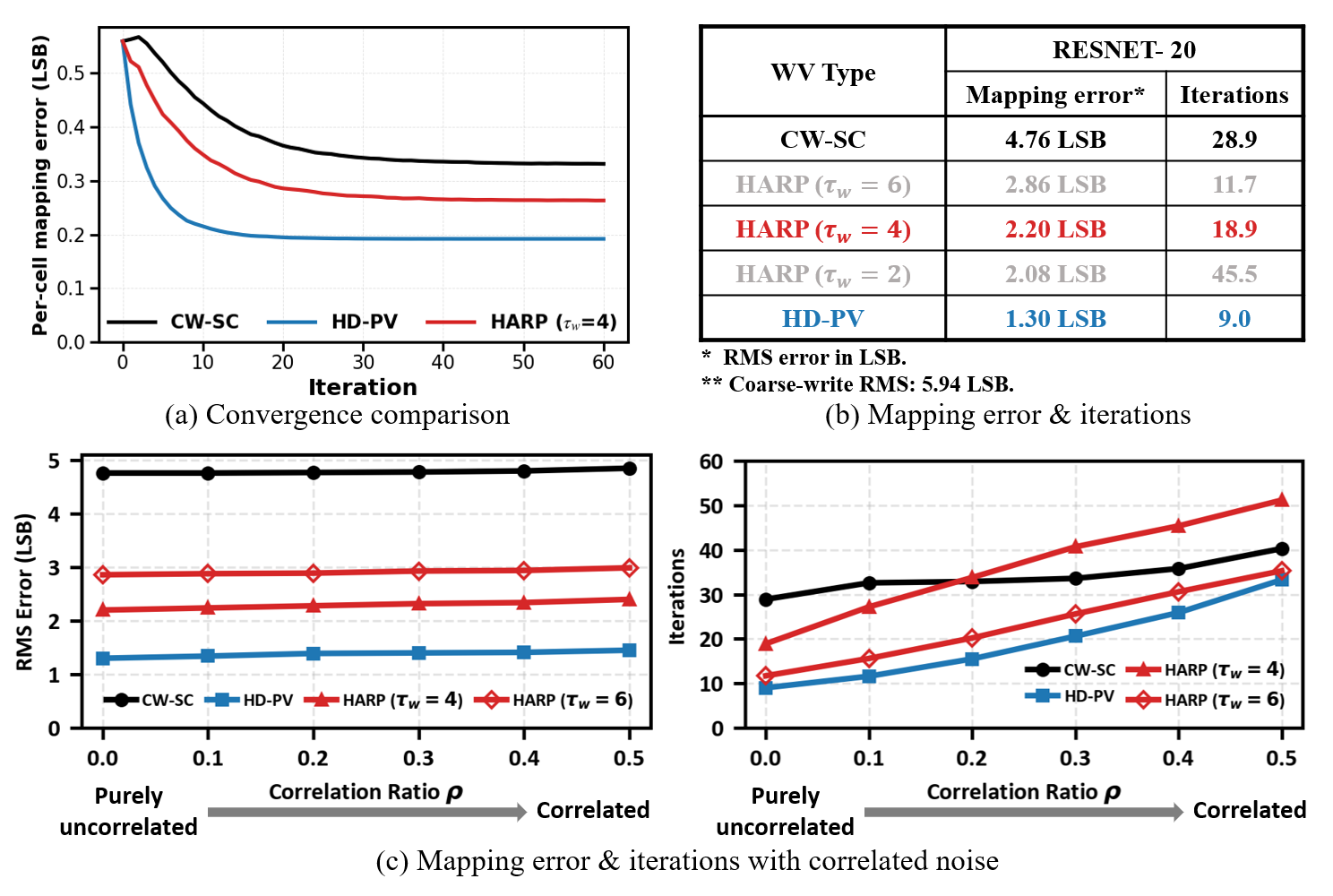}
    \vspace{-2em}
    \caption{Convergence comparison of CW-SC, HD-PV, and HARP under noisy verification. Hadamard-domain verification reduces both mapping error and iteration count.}
    \label{fig:graph_convergence}
    \vspace{-10pt}
\end{figure}

\subsection{WV Performance Comparison}

\begin{figure*}[t]
    \centering
    \includegraphics[width=0.95\textwidth]{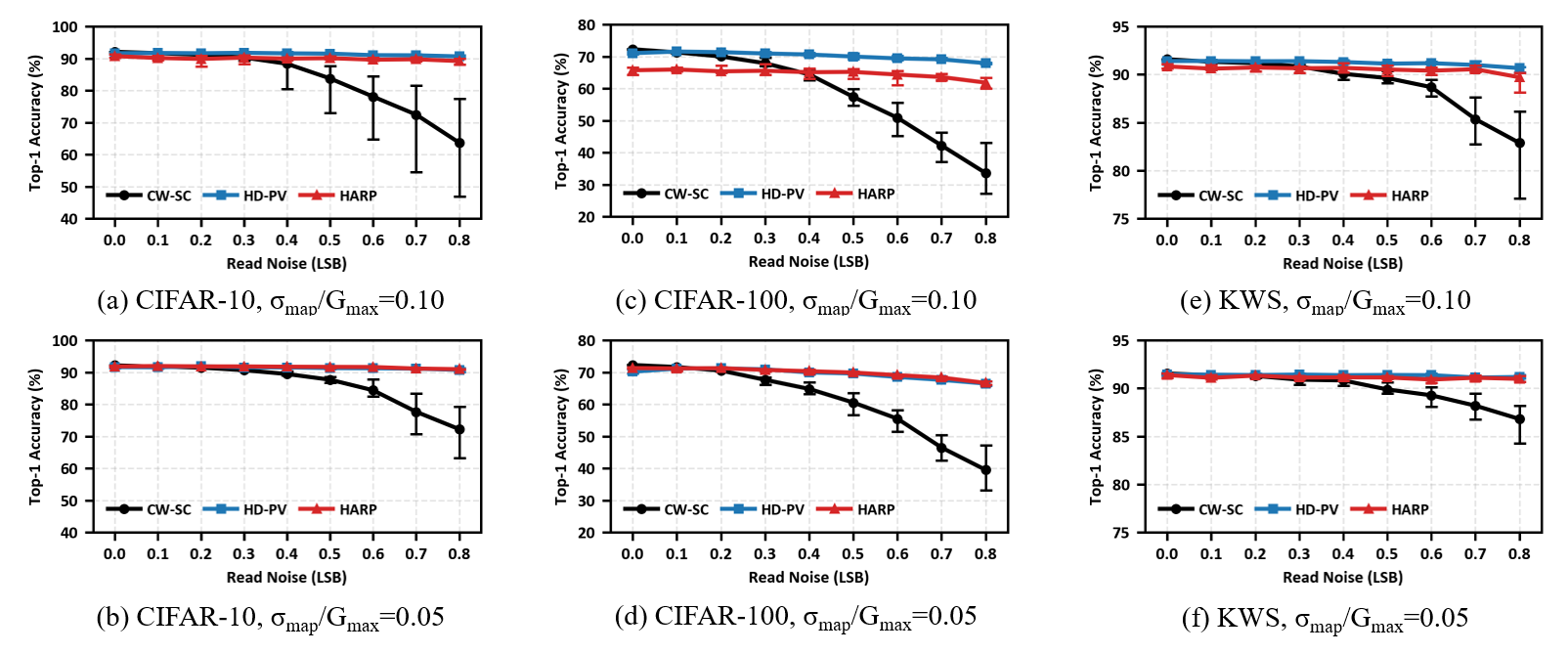}
    \vspace{-1em}
    \caption{Inference accuracy of CW-SC, HD-PV, and HARP on CIFAR-10, CIFAR-100, and KWS under two mapping noise levels ($\sigma_{\mathrm{map}}/G_{\max}=0.10$ and $0.05$). All results use the same configurations: \(B{=}6\), \(B_C{=}3\), \(N{=}32\), \(K{=}2\) and a 9-bit ADC.
}
    \label{fig:Accuracy}
    \vspace{-1em}
\end{figure*}

The default setting for the convergence study is \(B{=}6\), \(B_C{=}3\), \(N{=}32\), \(K{=}2\), mapping noise \(\sigma_{\mathrm{map}}/G_{\max}=0.10\), and read noise of 0.7\,LSB.
Update comparisons use a threshold of 0.5\,LSB.

Fig.~\ref{fig:graph_convergence}(a) compares the convergence behavior of CW-SC, HD-PV, and HARP. 
HD-PV achieves the steepest early error reduction because the inverse Hadamard step suppresses read noise before cell-domain decisions.
HARP converges more slowly than HD-PV because it operates on ternary comparison outcomes rather than full-valued readbacks, but it still clearly outperforms CW-SC.

The final mapping results in Fig.~\ref{fig:graph_convergence}(b) show the same trend quantitatively. 
HD-PV reduces the RMS mapping error from 4.76\,LSB to 1.30\,LSB, while cutting the iteration count from 28.9 to 9.0. 
HARP reduces the RMS error to 2.20\,LSB and the iteration count to 18.9.
The $\tau_w$ sensitivity study shows the expected accuracy--latency trade-off:
larger thresholds reduce WV iterations but may freeze cells earlier with larger mapping error, whereas smaller thresholds improves programming accuracy at the cost of more iterations.
We therefore use $\tau_w=4$ for the 32$\times$32 experiments and $\tau_w=6$ for the 64$\times$64 experiments as balanced operating points, while treating $\tau_w$ as a tunable knob rather than a fixed algorithmic constant.


Fig.~\ref{fig:graph_convergence}(c) shows the impact of common-mode noise. 
The total read noise power is fixed such that $\sqrt{\sigma_{\text{uc}}^2 + \sigma_{\text{cm}}^2} = 0.7$\,LSB, while $\rho = \sigma_{\text{cm}}^2 / (\sigma_{\text{uc}}^2 + \sigma_{\text{cm}}^2)$ is swept from $0$ (purely uncorrelated) to $0.5$ (equal correlated and uncorrelated power). 
Both HD-PV and HARP consistently achieve lower mapping error and fewer iterations than CW-SC across the entire range. 
Two mechanisms explain this: the uncorrelated component is reduced by $1/N$, while the common-mode component $\mu_\text{cm}$ is canceled for $N-1$ cells.
As $\rho$ increases, the second mechanism becomes dominant.
Multi-read averaging cannot cancel $\mu_\text{cm}$ because repeated reads share the same TIA/ADC.

\ignore{
To examine the sensitivity of HARP to the update threshold parameter $\tau_w$, Fig.~\ref{fig:graph_convergence}(b) also reports results for $\tau_w=6$ and $\tau_w=2$ (shown in gray). When $\tau_w$ is too large (e.g., $\tau_w=6$), updates terminate prematurely, leading to relatively larger mapping errors. Conversely, when $\tau_w$ is too small (e.g., $\tau_w=2$), the algorithm requires significantly more iterations before all cells satisfy the stopping condition, while the resulting RMS error improves only slightly compared to $\tau_w=4$. These results indicate that $\tau_w=4$ provides a practical balance between mapping accuracy and convergence speed, and therefore is used as the default configuration in the following experiments.
The larger RMS values in Fig.~\ref{fig:graph_convergence}(b), compared to the per-cell LSB errors in Fig.~\ref{fig:graph_convergence}(a), arise because each weight is represented by multiple binary-weighted cells, so MSB-cell errors dominate the aggregate deviation. Overall, Hadamard-based verification yields both faster convergence and more accurate programmed weights, with HD-PV providing the highest precision and HARP achieving strong accuracy with fewer iterations than CW-SC.
}

\subsection{Inference Accuracy under Iso-Memory-Footprint}

\begin{figure}[t]
    \centering
    \includegraphics[width=1.0\columnwidth]{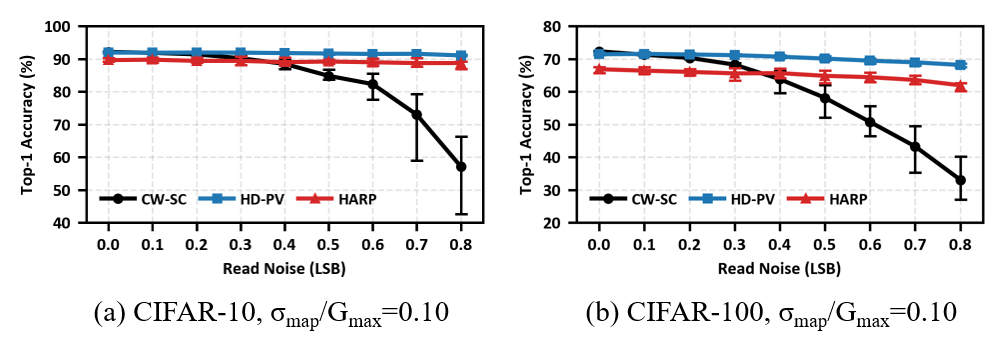}
    \vspace{-2em}
    \caption{CIFAR-10/100 accuracy for a 64$\times$64 array with \(B{=}6\), \(B_C{=}3\), \(N{=}64\), \(K{=}2\) and a 10-bit ADC. Hadamard-based verification remains robust even when the array and ADC resolution are scaled up.}
    \label{fig:graph_64x64}
    \vspace{-1em}
\end{figure}

We evaluate system-level inference on three edge-AI workloads: CIFAR-10 with ResNet-20, CIFAR-100 with ResNet-56, and KWS with a 12-class 10-layer CNN on the Google Speech Commands Dataset. All evaluations use identical bit precision, slice count, and array configuration, so comparisons are iso-memory-footprint:
any accuracy improvement comes from more reliable weight programming, not from additional memory.

Fig.~\ref{fig:Accuracy} shows the effect of read noise. 
CW-SC maintains accuracy only up to 0.2~LSB noise for CIFAR-10/100 and 0.4~LSB for KWS, then degrades rapidly, with over 20\,\% loss on CIFAR-10 at $\sigma \!\approx\! 0.8$~LSB and even larger loss on CIFAR-100.
HD-PV and HARP remain robust across the entire range, with less than 3\,\% degradation on all datasets.
Scalability to larger 64$\times$64 arrays (Fig.~\ref{fig:graph_64x64}) confirms the same trend: 
the denoising benefit of Hadamard verification grows with column length (both the $1/N$ variance reduction and $N-1$ common-mode-canceled cells scale favorably),
so HD-PV and HARP maintain near-constant accuracy where CW-SC becomes highly vulnerable.
Multi-read averaging~\cite{joshi2020accurate,shim2020two} can partially recover accuracy (Fig.~\ref{fig:graph_mult}), but at substantial ADC overhead analyzed in Section~\ref{sec:latency_energy}.

\begin{figure}[t]
    \centering
    \includegraphics[width=0.9\columnwidth]{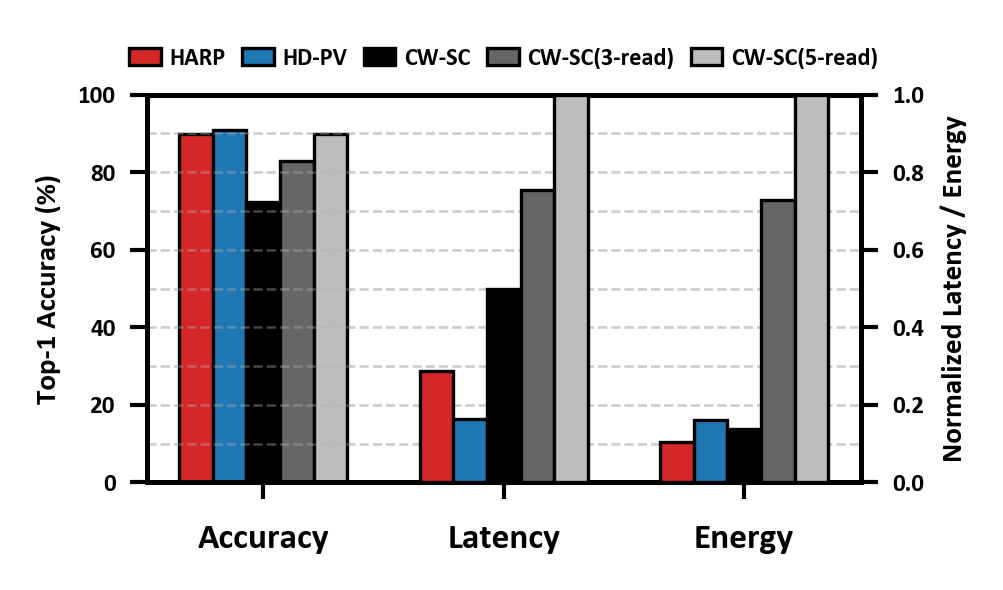}
    \vspace{-1em}
    \caption{Accuracy, normalized latency, and normalized energy on CIFAR-10 under identical memory footprint with $\sigma_{\mathrm{map}}/G_{\max}=0.10$ and read noise of $0.7$\,LSB with \(B{=}6\), \(B_C{=}3\), \(N{=}32\), \(K{=}2\) and a 9-bit ADC. HD-PV and HARP match the robustness of multi-read averaging with much lower ADC overhead.}
    \label{fig:graph_mult}
    \vspace{-1.5em}
\end{figure}


\subsection{Latency and Energy}
\label{sec:latency_energy}

\begin{figure}[t]
    \centering
    \includegraphics[width=0.95\columnwidth]{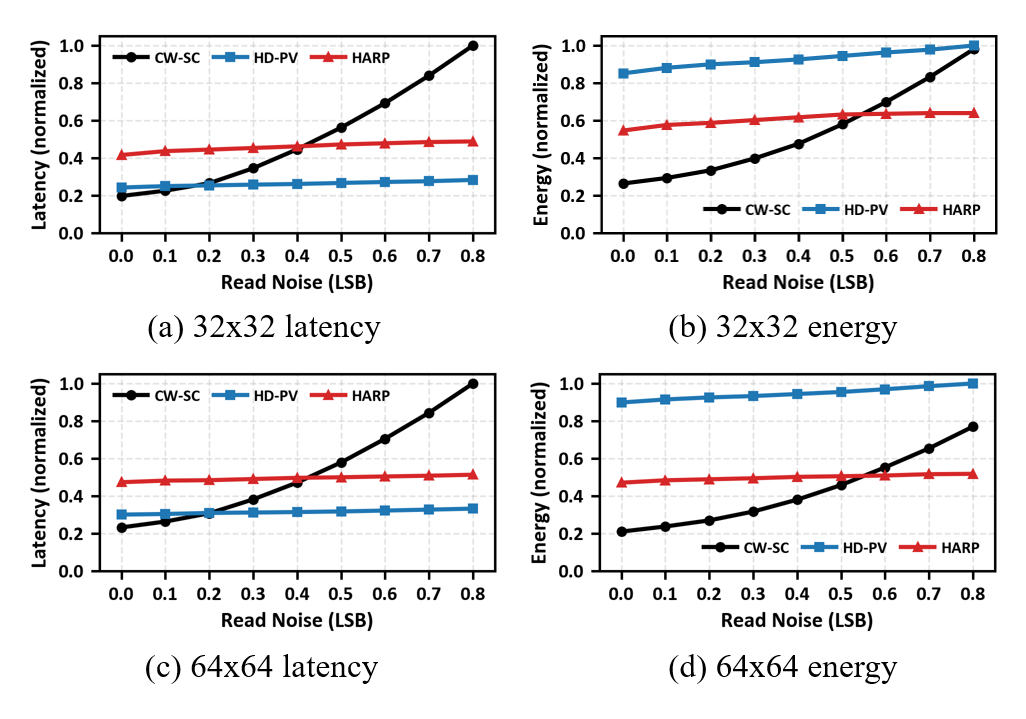}
    \vspace{-1em}
    \caption{WV latency and energy for CIFAR-10 (ResNet-20) versus read noise, with $\sigma_{\mathrm{map}}/G_{\max}=0.10$, $B{=}6$, $B_C{=}3$ and \(K{=}2\). (a)–(b) for a $32\times32$ array with a 9-bit ADC, and (c)–(d) for a $64\times64$ array with a 10-bit ADC. HARP provides the best latency-energy trade-off under noisy verification.}
    \label{fig:Latency_Energy}
    \vspace{-1.5em}
\end{figure}

Fig.~\ref{fig:graph_mult} highlights the efficiency of the proposed methods over multi-read averaging.
Although 5-read averaging recovers CW-SC accuracy, HD-PV and HARP are $6.1\times$ and $3.5\times$ faster and $6.2\times$ and $9.5\times$ more energy-efficient, respectively.
This improvement stems from the large ADC overhead incurred by multi-read averaging, which requires repeated full conversions for each read operation.

Fig.~\ref{fig:Latency_Energy} reports per-column WV latency and energy for CIFAR-10 (ResNet-20).
For the 32$\times$32 array, CW-SC is competitive at very low noise ($\leq$0.1~LSB) because it avoids Hadamard decoding, but its latency grows rapidly once noisy verify reads trigger incorrect update decisions and extra iterations. 
Above 0.4~LSB, CW-SC becomes the slowest method. 
HD-PV and HARP exhibit only modest latency growth (16\,\% and 17\,\% for 32$\times$32, and 9.7\,\% and 8.9\,\% for 64$\times$64) because they maintain stable decisions as noise increases.

Energy shows a different ordering. HD-PV has the highest per-read cost because every Hadamard measurement uses a full SAR conversion. 
HARP instead combines Hadamard denoising with compare-only ADC operation and consistently achieves the lowest energy in the high-noise regime. 
Under severe noise ($\geq$0.5~LSB) HARP uses $\approx$65\,\% of HD-PV’s energy for 32$\times$32.

Across all configurations, ADC activity accounts for more than 70\,\% of WV latency and more than 90\,\% of WV energy. 
This explains why reducing both misdirected iterations (via Hadamard denoising) and per-read ADC work (via compare-only mode) is crucial.

\subsection{Comparison with Prior Work}

Table~\ref{tab:comparison} compares the proposed schemes against representative WV studies.
The comparison highlights three levels of differentiation.

{\bf Optimization target.}
Prior works~\cite{gonugondla2020swipe,meng2022write,huang2024rwric, yan2022swim, zhang2021efficient} focus on the write policy, optimizing decisions under a fixed observation quality.
HD-PV and HARP instead optimize the observation itself---the verify read---by changing the measurement basis to the Hadamard basis.
The two approaches are complementary: better observations enable better decisions regardless of the pulse scheduling strategy.

{\bf Circuit-level cost modeling.}
Of the prior methods, only \cite{gonugondla2020swipe} specifies the ADC and evaluates at specific array dimensions.
References \cite{meng2022write,huang2024rwric, yan2022swim, zhang2021efficient} do not model ADC latency or energy, making it difficult to assess the true WV overhead.
HD-PV and HARP include explicit circuit-level modeling with practical ADC parameters.

{\bf ADC-aware verification.}
HARP is the only method that redesigns the ADC operating mode for WV.
By recognizing that WV requires classification rather than full digitization, HARP avoids most SAR work---a unique contribution not addressed by any prior method.
Unlike prior iteration-count estimates, our results use circuit-level simulations with practical parameters.

All improvements in Table~\ref{tab:comparison} are normalized against the CW-SC WV baseline---already superior to conventional cell-by-cell WV---indicating that the actual gains of HD-PV and HARP over conventional WV are even larger than the ratios reported.

\begin{table}[t]
\centering
\caption{Comparison with prior work.}
\label{tab:comparison}
\vspace{-1.5em}
\renewcommand{\arraystretch}{1.12}
\setlength{\tabcolsep}{2pt}
\scriptsize
\resizebox{\columnwidth}{!}{%
\begin{tabular}{lccc|cc}
\toprule
\textbf{Metric} & \textbf{ICCAD'20~\cite{gonugondla2020swipe}} & \textbf{DAC'22~\cite{meng2022write}} & \textbf{DAC'24~\cite{huang2024rwric}} & \textbf{HD-PV} & \textbf{HARP} \\
\midrule
Weight Precision         & 2--9 bit & 8 bit & 8 bit & \multicolumn{2}{c}{4--6 bit} \\
MLC                      & 1--3 bit & 2 bit & 2--4 bit & \multicolumn{2}{c}{2--3 bit} \\
ADC Specification        & Yes (7--10 bit) & No & No & \multicolumn{2}{c}{Yes (8--10 bit)} \\
DNN Dataset              & \makecell[c]{CIFAR-10\\MNIST}
                         & \makecell[c]{CIFAR-10/100}
                         & \makecell[c]{CIFAR-10/100\\COCO}
                         & \multicolumn{2}{c}{\makecell[c]{CIFAR-10/100\\KWS}} \\
Network Type             & \makecell[c]{LeNet-300/100\\8-layer CNN}
                         & \makecell[c]{ResNet-18/34}
                         & \makecell[c]{ResNet-50\\YOLOv8x}
                         & \multicolumn{2}{c}{\makecell[c]{ResNet-20/56\\10-layer CNN}} \\
Accuracy Change          & $<$1\% drop & 0.23\% gain & 0.9\% drop & 1.08\% gain & 0.16\% drop \\
Energy Change            & 5--10$\times$ & 10.3$\times$ & -- & 6.2$\times$ & 9.5$\times$ \\
Latency Change           & -- & -- & -- & 6.1$\times$ & 3.5$\times$ \\
Array Dimension          & 32$\times$32 & -- & -- & \multicolumn{2}{c}{32$\times$32/64$\times$64} \\
\bottomrule
\end{tabular}
}
\normalsize
\vspace{-1em}
\end{table}

\section{Conclusion}

This work presented HD-PV and HARP, two circuit-aware WV schemes for the verify-read bottleneck in RRAM-based ACiM.
HD-PV replaces statistically suboptimal one-hot verification with Hadamard encoding, reducing uncorrelated-noise variance by $1/N$ and canceling common-mode disturbances for $N-1$ cells without additional analog hardware.
By treating WV as classification rather than estimation, HARP replaces full ADC conversions with compare-only decisions.
Under an iso-memory footprint, both methods recover accuracy lost to verify-read noise, while the improved programming margin can reduce bit slices or array size at fixed accuracy.
The framework is transparent to DNN models and training, complementary to pulse scheduling, and compatible with existing macros.

\vspace{-1em}

\begin{acks}
    This work was supported in part by K-CHIPS(2410018409, RS-2026-25524122, 26086-45FC) funded by MOTIE, by NRF (RS-2024-00402495), and by IITP(IITP-2026-RS-2023-00256081) grant funded by the Korea government(MSIT).
    The EDA tool was supported by the IC Design Education Center(IDEC), Korea.
    Woo-Seok Choi is the corresponding author.
\end{acks}

\bibliographystyle{ACM-Reference-Format}
\bibliography{refs}

@article{chen2018neurosim,
  title={NeuroSim: A circuit-level macro model for benchmarking neuro-inspired architectures in online learning},
  author={Chen, Pai-Yu and Peng, Xiaochen and Yu, Shimeng},
  journal={IEEE Transactions on Computer-Aided Design of Integrated Circuits and Systems},
  volume={37},
  number={12},
  pages={3067--3080},
  year={2018},
  publisher={IEEE}
}

@article{hung2020challenges,
  title={Challenges and trends indeveloping nonvolatile memory-enabled computing chips for intelligent edge devices},
  author={Hung, Je-Min and Li, Xueqing and Wu, Juejian and Chang, Meng-Fan},
  journal={IEEE Transactions on Electron Devices},
  volume={67},
  number={4},
  pages={1444--1453},
  year={2020},
  publisher={IEEE}
}

@article{yao2020fully,
  title={Fully hardware-implemented memristor convolutional neural network},
  author={Yao, Peng and Wu, Huaqiang and Gao, Bin and Tang, Jianshi and Zhang, Qingtian and Zhang, Wenqiang and Yang, J Joshua and Qian, He},
  journal={Nature},
  volume={577},
  number={7792},
  pages={641--646},
  year={2020},
  publisher={Nature Publishing Group UK London}
}

@article{wan2022compute,
  title={A compute-in-memory chip based on resistive random-access memory},
  author={Wan, Weier and Kubendran, Rajkumar and Schaefer, Clemens and Eryilmaz, Sukru Burc and Zhang, Wenqiang and Wu, Dabin and Deiss, Stephen and Raina, Priyanka and Qian, He and Gao, Bin and others},
  journal={Nature},
  volume={608},
  number={7923},
  pages={504--512},
  year={2022},
  publisher={Nature Publishing Group UK London}
}

@article{sze2017efficient,
  title={Efficient processing of deep neural networks: A tutorial and survey},
  author={Sze, Vivienne and Chen, Yu-Hsin and Yang, Tien-Ju and Emer, Joel S},
  journal={Proceedings of the IEEE},
  volume={105},
  number={12},
  pages={2295--2329},
  year={2017},
  publisher={Ieee}
}

@inproceedings{gonugondla2020swipe,
  title={Swipe: Enhancing robustness of reram crossbars for in-memory computing},
  author={Gonugondla, Sujan K and Patil, Ameya D and Shanbhag, Naresh R},
  booktitle={Proceedings of the 39th International Conference on Computer-Aided Design},
  pages={1--9},
  year={2020}
}

@inproceedings{meng2022write,
  title={Write or not: Programming scheme optimization for RRAM-based neuromorphic computing},
  author={Meng, Ziqi and Sun, Yanan and Qian, Weikang},
  booktitle={Proceedings of the 59th ACM/IEEE Design Automation Conference},
  pages={985--990},
  year={2022}
}

@inproceedings{he2023prive,
  title={Prive: efficient RRAM programming with chip verification for RRAM-based in-memory computing acceleration},
  author={He, Wangxin and Meng, Jian and Gonugondla, Sujan Kumar and Yu, Shimeng and Shanbhag, Naresh R and Seo, Jae-sun},
  booktitle={2023 Design, Automation \& Test in Europe Conference \& Exhibition (DATE)},
  pages={1--6},
  year={2023},
  organization={IEEE}
}

@article{yu2021rram,
  title={RRAM for compute-in-memory: From inference to training},
  author={Yu, Shimeng and Shim, Wonbo and Peng, Xiaochen and Luo, Yandong},
  journal={IEEE Transactions on Circuits and Systems I: Regular Papers},
  volume={68},
  number={7},
  pages={2753--2765},
  year={2021},
  publisher={IEEE}
}

@inproceedings{huang2024rwric,
  title={RWriC: A Dynamic Writing Scheme for Variation Compensation for RRAM-based In-Memory Computing},
  author={Huang, Yucong and He, Jingyu and Cheng, Tim Kwang-Ting and Tsui, Chi Ying and Ye, Terry Tao},
  booktitle={Proceedings of the 61st ACM/IEEE Design Automation Conference},
  pages={1--6},
  year={2024}
}

@inproceedings{zhang2021efficient,
  title={An efficient programming framework for memristor-based neuromorphic computing},
  author={Zhang, Grace Li and Li, Bing and Huang, Xing and Shen, Chen and Zhang, Shuhang and Burcea, Florin and Graeb, Helmut and Ho, Tsung-Yi and Li, Hai and Schlichtmann, Ulf},
  booktitle={2021 Design, Automation \& Test in Europe Conference \& Exhibition (DATE)},
  pages={1068--1073},
  year={2021},
  organization={IEEE}
}

@article{correll20258,
  title={An 8-bit 20.7 TOPS/W Multilevel Cell ReRAM Macro With ADC-Assisted Bit-Serial Processing},
  author={Correll, Justin M and Jie, Lu and Song, Seungheun and Lee, Seungjong and Zhu, Junkang and Tang, Wei and Wormald, Luke and Erhardt, Jack and Breil, Nicolas and Quon, Roger and others},
  journal={IEEE Journal of Solid-State Circuits},
  year={2025},
  publisher={IEEE}
}

@inproceedings{he2019noise,
  title={Noise injection adaption: End-to-end ReRAM crossbar non-ideal effect adaption for neural network mapping},
  author={He, Zhezhi and Lin, Jie and Ewetz, Rickard and Yuan, Jiann-Shiun and Fan, Deliang},
  booktitle={Proceedings of the 56th Annual Design Automation Conference 2019},
  pages={1--6},
  year={2019}
}

@article{khashaba2019low,
  title={A low-noise frequency synthesizer using multiphase generation and combining techniques},
  author={Khashaba, Amr and Elkholy, Ahmed and Megawer, Karim M and Ahmed, Mostafa Gamal and Hanumolu, Pavan Kumar},
  journal={IEEE Journal of Solid-State Circuits},
  volume={55},
  number={3},
  pages={592--601},
  year={2019},
  publisher={IEEE}
}

@article{joshi2020accurate,
  title={Accurate deep neural network inference using computational phase-change memory},
  author={Joshi, Vinay and Le Gallo, Manuel and Haefeli, Simon and Boybat, Irem and Nandakumar, Sasidharan Rajalekshmi and Piveteau, Christophe and Dazzi, Martino and Rajendran, Bipin and Sebastian, Abu and Eleftheriou, Evangelos},
  journal={Nature communications},
  volume={11},
  number={1},
  pages={2473},
  year={2020},
  publisher={Nature Publishing Group UK London}
}

@article{ielmini2025resistive,
  title={Resistive switching random-access memory (RRAM): Applications and requirements for memory and computing},
  author={Ielmini, Daniele and Pedretti, Giacomo},
  journal={Chemical Reviews},
  year={2025},
  publisher={ACS Publications}
}

@article{roh2023context,
  title={A context-aware readout system for sparse touch sensing array using ultra-low-power always-on event detection},
  author={Roh, Hyeri and Choi, Woo-Seok},
  journal={IEEE Transactions on Circuits and Systems II: Express Briefs},
  volume={70},
  number={9},
  pages={3719--3723},
  year={2023},
  publisher={IEEE}
}

@article{xie2023high,
  title={A high-parallelism RRAM-based compute-in-memory macro with intrinsic impedance boosting and in-ADC computing},
  author={Xie, Tian and Yu, Shimeng and Li, Shaolan},
  journal={IEEE Journal on Exploratory Solid-State Computational Devices and Circuits},
  volume={9},
  number={1},
  pages={38--46},
  year={2023},
  publisher={IEEE}
}

@article{li202240,
  title={A 40-nm MLC-RRAM compute-in-memory macro with sparsity control, on-chip write-verify, and temperature-independent ADC references},
  author={Li, Wantong and Sun, Xiaoyu and Huang, Shanshi and Jiang, Hongwu and Yu, Shimeng},
  journal={IEEE Journal of Solid-State Circuits},
  volume={57},
  number={9},
  pages={2868--2877},
  year={2022},
  publisher={IEEE}
}

@article{tang2022low,
  title={Low-power SAR ADC design: Overview and survey of state-of-the-art techniques},
  author={Tang, Xiyuan and Liu, Jiaxin and Shen, Yi and Li, Shaolan and Shen, Linxiao and Sanyal, Arindam and Ragab, Kareem and Sun, Nan},
  journal={IEEE Transactions on Circuits and Systems I: Regular Papers},
  volume={69},
  number={6},
  pages={2249--2262},
  year={2022},
  publisher={IEEE}
}

@article{peng2020dnn+,
  title={DNN+ NeuroSim V2. 0: An end-to-end benchmarking framework for compute-in-memory accelerators for on-chip training},
  author={Peng, Xiaochen and Huang, Shanshi and Jiang, Hongwu and Lu, Anni and Yu, Shimeng},
  journal={IEEE Transactions on Computer-Aided Design of Integrated Circuits and Systems},
  volume={40},
  number={11},
  pages={2306--2319},
  year={2020},
  publisher={IEEE}
}

@inproceedings{andrulis2023raella,
  title={Raella: Reforming the arithmetic for efficient, low-resolution, and low-loss analog pim: No retraining required!},
  author={Andrulis, Tanner and Emer, Joel S and Sze, Vivienne},
  booktitle={Proceedings of the 50th Annual International Symposium on Computer Architecture},
  pages={1--16},
  year={2023}
}

@inproceedings{yan2022swim,
  title={Swim: Selective write-verify for computing-in-memory neural accelerators},
  author={Yan, Zheyu and Hu, Xiaobo Sharon and Shi, Yiyu},
  booktitle={Proceedings of the 59th ACM/IEEE Design Automation Conference},
  pages={277--282},
  year={2022}
}

@INPROCEEDINGS{9870009,
  author={Song, Joonghyun and Shin, Jiwon and Kim, Hanseok and Choi, Woo-Seok},
  booktitle={2022 IEEE 4th International Conference on Artificial Intelligence Circuits and Systems (AICAS)}, 
  title={Energy-Efficient High-Accuracy Spiking Neural Network Inference Using Time-Domain Neurons}, 
  year={2022},
  volume={},
  number={},
  pages={5-8},
  doi={10.1109/AICAS54282.2022.9870009}}

@inproceedings{horowitz20141,
  title={1.1 computing's energy problem (and what we can do about it)},
  author={Horowitz, Mark},
  booktitle={2014 IEEE international solid-state circuits conference digest of technical papers (ISSCC)},
  pages={10--14},
  year={2014},
  organization={IEEE}
}

@article{shim2020two,
  title={Two-step write--verify scheme and impact of the read noise in multilevel RRAM-based inference engine},
  author={Shim, Wonbo and Seo, Jae-sun and Yu, Shimeng},
  journal={Semiconductor Science and Technology},
  volume={35},
  number={11},
  pages={115026},
  year={2020},
  publisher={IOP Publishing}
}

@book{harwit2012hadamard,
  title={Hadamard transform optics},
  author={Harwit, Martin},
  year={2012},
  publisher={Elsevier}
}

@article{hotelling1944some,
  title={Some improvements in weighing and other experimental techniques},
  author={Hotelling, Harold},
  journal={The Annals of Mathematical Statistics},
  volume={15},
  number={3},
  pages={297--306},
  year={1944},
  publisher={JSTOR}
}

@article{zhang2023edge,
  title={Edge learning using a fully integrated neuro-inspired memristor chip},
  author={Zhang, Wenbin and Yao, Peng and Gao, Bin and Liu, Qi and Wu, Dong and Zhang, Qingtian and Li, Yuankun and Qin, Qi and Li, Jiaming and Zhu, Zhenhua and others},
  journal={Science},
  volume={381},
  number={6663},
  pages={1205--1211},
  year={2023},
  publisher={American Association for the Advancement of Science}
}

@article{wu2023bulk,
  title={Bulk-switching memristor-based compute-in-memory module for deep neural network training},
  author={Wu, Yuting and Wang, Qiwen and Wang, Ziyu and Wang, Xinxin and Ayyagari, Buvna and Krishnan, Siddarth and Chudzik, Michael and Lu, Wei D},
  journal={Advanced Materials},
  volume={35},
  number={46},
  pages={2305465},
  year={2023},
  publisher={Wiley Online Library}
}

@inproceedings{spetalnick20232,
  title={A 2.38 mcells/mm 2 9.81-350 tops/W RRAM compute-in-memory macro in 40nm CMOS with hybrid offset/I OFF cancellation and I CELL R BLSL drop mitigation},
  author={Spetalnick, Samuel D and Chang, Muya and Konno, Shota and Crafton, Brian and Lele, Ashwin S and Khwa, Win-San and Chih, Yu-Der and Chang, Meng-Fan and Raychowdhury, Arijit},
  booktitle={2023 IEEE Symposium on VLSI Technology and Circuits (VLSI Technology and Circuits)},
  pages={1--2},
  year={2023},
  organization={IEEE}
}

@article{giacomin2018resistive,
  title={A resistive random access memory addon for the ncsu freepdk 45 nm},
  author={Giacomin, Edouard and Gaillardon, Pierre-Emmanuel},
  journal={IEEE Transactions on Nanotechnology},
  volume={18},
  pages={68--72},
  year={2018},
  publisher={IEEE}
}

\end{document}